\documentclass[oneside, 11pt, reqno]{amsart}
\usepackage[left=1in, right=1in, top=1in, bottom=1in]{geometry}

\usepackage[usenames,dvipsnames]{xcolor}
\usepackage{amsmath}

\usepackage{graphicx,enumitem}
\usepackage{mathtools}

\usepackage{color}
\usepackage[dvipsnames]{xcolor}

\definecolor{myblue}{RGB}{94, 129, 181}
\definecolor{myorange}{RGB}{225, 156, 36}
\definecolor{mygreen}{RGB}{143, 176, 50}

\usepackage{amssymb} 
\usepackage{soul}  

\usepackage[margin=0.7cm]{caption}
\usepackage{multirow}

\newcommand{\bR}{{\mathbb R}}

\newcommand{\PP}{{\mathbb P}}
\newcommand{\bE}{{\mathbb E}}

\newcommand{\FF}{{\mathcal F}}
\newcommand{\cA}{{\mathcal A}}

\newcommand{\Zitkovic}[1]{{\v Z}itkovi\'c}
\newcommand{\Sirbu}[1]{S\^\i rbu}

\newcommand\tv{{\tilde{v}}}
\newcommand\ta{{\tilde{a}}}

\numberwithin{equation}{section}
\theoremstyle{plain}                
\newtheorem{theorem}{Theorem}[section]
\newtheorem{lemma}[theorem]{Lemma}
\newtheorem{proposition}[theorem]{Proposition}

\theoremstyle{definition}           
\newtheorem{definition}[theorem]{Definition}

\theoremstyle{remark}

\newcommand{\defref}[1]{Definition~\ref{#1}}
\newcommand{\thmref}[1]{Theorem~\ref{#1}}
\newcommand{\proref}[1]{Proposition~\ref{#1}}

\newcommand{\lemref}[1]{Lemma~\ref{#1}}

\usepackage{graphicx} 
\usepackage{xcolor}

\begin{document}

\author{Jin Hyuk Choi$^\dag$, Heeyoung Kwon$^\dag$, and Kasper Larsen$^\ddag$}

\thanks{The authors have benefited from helpful comments from two anonymous referees and Umut  \c{C}et\.in.}

\thanks{$^\dag$Ulsan National Institute of Science and Technology (jchoi@unist.ac.kr, heeee02@unist.ac.kr).
The research of these authors was supported by the National Research Foundation of Korea (NRF) grant funded by the Korea government (MSIT) (No. 2020R1C1C1A01014142 and No. 2021R1A4A1032924).}
\thanks{$^\ddag$ Rutgers University (KL756@math.rutgers.edu). The research of this author was supported by the National Science Foundation under Grant No. DMS 1812679 (2018 - 2022). Any opinions, findings, and conclusions or recommendations expressed in this material are those of the author(s) and do not necessarily reflect the views of the National Science Foundation (NSF)}

\title{Trading constraints in continuous-time Kyle models}

\begin{abstract} 
In a continuous-time Kyle setting, we prove global existence of an equilibrium when the insider  faces a terminal trading constraint. We prove that our equilibrium model produces output consistent with several empirical stylized facts such as autocorrelated aggregate holdings, decreasing price impacts over the trading day, and $U$ shaped optimal trading patterns.
\end{abstract}

\maketitle


\section{Introduction}

Kyle (1985) and Back (1992) are cornerstone works in market microstructure theory and we develop a continuous-time trading model variation where the insider faces a trading constraint. The constraint is modeled by a random variable $\ta$ which is private information for the insider. We use a ``hard" constraint formulation in the sense that the insider's stock holding process $(\theta_t)_{t\in[0,T]}$ must be such that its terminal value $\theta_T$ satisfies $\theta_T=\ta$ where $T\in(0,\infty)$ is the trading horizon. Our main theorem ensures the existence of a Kyle  equilibrium when the insider maximizes her expected profit subject to $\theta_T=\ta$. We prove that adding a trading restriction produces new model outputs relative to Kyle (1985) and Back (1992) such as: 

\begin{itemize}
\item[(i)] Positive autocorrelation in aggregate holdings $Y_t :=\theta_t+\sigma_wW_t$ where $\sigma_wW_t$ denotes the exogenous noise trader holdings ($W_t$ is a Brownian motion) and $\sigma_w>0$ is a constant.\footnote The process $Y_t$ is also  referred to as cumulative aggregate orders because,  informally, $dY_t$ are aggregate orders.
\item[(ii)] A decreasing price impact function $\lambda(t)$ for $t\in[0,T]$ (commonly called Kyle's $\lambda$). \item[(iii)]  $U$ shaped optimal stock orders for the insider (in expectation). 
\end{itemize}

These three features are empirically desirable (see, e.g., Barardehi and Bernhardt (2021)). While absent in Kyle (1985) and Back (1992), related Kyle models have produced similar features. For example, when the time horizon is exponentially distributed, Back and Baruch (2004), Caldentey and Stacchetti (2010), and \c{C}et\.in (2018) show that $\lambda$ is a supermartingale. Baruch (2002) and Cho (2003) show that an insider with an exponential utility function can also produce a decreasing $\lambda$. In contrast, Collin-Dufresne and Vyacheslav (2016) introduce stochastic noise-trader volatility and show this feature makes $\lambda$ a submartingale. \c{C}et\.in and Danilova (2016) show that market makers with exponential utilities can produce mean reversion in the equilibrium aggregate holdings as well as non-martingality of $\lambda$.

While trading restrictions have been used in many optimal investment problems, see, e.g., Almgren and Chriss (1999, 2000), Almgren (2003), and Schied and Sch\"oneborn (2009), restrictions produce model incompleteness and this complicates any equilibrium analysis significantly. Equilibrium existence proofs are given in: (i) Radner equilibrium models with limited participation in Basak and Cuoco (1998), Hugonnier (2012), and Prieto (2010). (ii) Nash equilibrium models with predatory trading in Brunnermeier and Pedersen (2005) and Carlin, Lobo, and Viswanathan (2007). (iii) Models based on ``soft" constraints in the form of quadratic penalties  incurred for deviating from the terminal constraint $\ta$ in G\^arleanu and Pedersen (2016), Bouchard, Fukasawa, Herdegen, and Muhle-Karbe (2018), and Choi, Larsen, and Seppi (2021).  Models based on ``soft'' constraints are mathematically simpler relative to the ``hard'' constraints used in (i), (ii), and in this paper. To the best of our knowledge, there is no equilibrium existence proofs in the settings of Kyle (1985) and Back (1992) when the insider is subject to either a ``hard'' or ``soft'' trading constraint. However, our paper is related to Degryse, de Jong, and van Kervel (2014),  Choi, Larsen, and Seppi (2019), van Kervel, Kwan, and Westerholm (2020)  who numerically consider Kyle models where there are both an unconstrained insider and an additional trader facing a terminal trading constraint. These papers do not give existence proofs.

Our main theorem establishes equilibrium existence by first proving existence of solutions to an autonomous system of first-order nonlinear ODEs and then provides the equilibrium stock-price and holding processes in terms of these solutions. Our  model has two mathematically critical components. First, the insider's terminal trading restriction $\theta_T = \ta$ requires a new state process in addition to the aggregate holding process $Y_t=\theta_t+\sigma_wW_t $ ($Y_t$   is the standard state process in many Kyle models). Consequently, the pricing rules given by  $H(t,Y_t)$ for a deterministic function $H$ used in both Kyle (1985) and Back (1992) are insufficient in our setting. We create generalized pricing rules based on a two dimensional state process and prove that such a generalization produces existence of a global Kyle equilibrium. Second, we show that our constrained insider optimally places a terminal block order $\Delta \theta_T\neq 0$. In a class of Kyle models, Back (1992) shows that block orders are always suboptimal. \c{C}et\.in and Danilova (2021) give general conditions under which block orders are suboptimal but they also illustrate that block orders can be optimal in more general Kyle models. Additionally, we show that the market makers can predict our insider's optimal terminal block order $\Delta \theta_T = \ta - \theta_{T-}$.

The paper is organized as follows: Section 2 recalls the continuous-time setting in  Back (1992) with Gaussian dividends. Section 3 ensures the existence of ODE solutions. Section 4 gives our main theorem. Section 5 contains results related to empirical stylized facts. Appendix \ref{app:KB} recalls the version of the Kalman-Bucy theorem from filtering theory we need and Appendix \ref{app:full} considers a model variation where the insider knows both $\ta$ and $\tv$ initially.

\bigskip

\section{Model}

Except for the insider's trading constraint, we use the continuous-time setup in  Back (1992) specialized to the Gaussian dividend case as in Kyle (1985). The time interval for trading is denoted by  $[0,T]$ for an  arbitrary time horizon $T\in (0,\infty)$. In the following, we use three sources of randomness: $(W_t)_{t\in[0,T]}$ is a one-dimensional Brownian motion with zero initial value, zero drift, and constant unit volatility, and the random variables $(\ta,\tv)$ are  jointly normally distributed with zeros means, standard deviations $\sigma_{\ta},\sigma_{\tv}>0$, and correlation $\rho \in (0,1]$. We assume that $(W_t)_{t\in[0,T]}$  is independent of $(\ta,\tv)$.

The financial market consists of a money market account with an exogenous zero interest rate and a stock with liquidating dividends given by $\tv$ and an endogenously determined stock-price process  $(P_t)_{t\in [0,T]}$. 

\medskip

\subsection{Traders}

There are three types of market participants in the model:  noise traders, an insider with a trading constraint, and market makers:

\bigskip

\noindent{\emph{Noise traders:} } For a constant $\sigma_w>0$, these traders' aggregate stock-holding process is exogenously given by $\sigma_w W_t$ at time $t\in[0,T]$. Equivalently, the noise-traders' exogenous initial holdings are zero and their holdings have dynamics $\sigma_w dW_t$ over time $t\in(0,T)$.

\bigskip
 
\noindent{\emph{Insider:} } The insider's stock-holding process is denoted by $(\theta_t)_{t\in[0,T]}$. The insider starts with zero initial stock holdings $\theta_{0-}:=0$ and is subject to the terminal trading constraint $\theta_T=\ta$. 
The insider knows $\ta$ at time $t=0$ and observes the noise-trader orders $\sigma_w dW_t$ over time; hence, the insider's filtration is:
\begin{align}\label{FI_t}
\FF_t^I:=\sigma\left(\ta, (W_u)_{u\in [0,t]}  \right)\quad \textrm{for} \quad t\in [0,T].
\end{align} 
Equivalently, we can follow Back (1992) and assume that the insider directly observes the aggregate orders $dY_t$ over time where 
\begin{align}\label{Y}
Y_t:=\theta_t + \sigma_wW_t,\quad t\in[0,T],
\end{align}
 and replace $\FF_t^I$ with $\sigma\left(\ta, (Y_u)_{u\in [0,t]}  \right)$. A second equivalent formulation assumes  that the insider observes stock prices $P_t$ over time and  replaces $\FF_t^I$ with $\sigma\left(\ta, (P_u)_{u\in [0,t]}  \right)$.  In our setting, these three specifications of the insider's filtration turn out to be equivalent and we use \eqref{FI_t} because it is exogenously given by model inputs.
 
 The random variable $\ta$ plays a double role: $\ta$ is the insider's terminal trading restriction and $\ta$ gives the insider initial private information about $\tv$ (because $\rho>0$). We are not the first authors to use holdings to generate private information. For example, the activism model in Back,  Collin-Dufresne, Fos, Li,  and Ljungqvist (2018) uses the insider's initial holdings as the key piece of private information.  However, in Appendix \ref{app:full}, we show that the equilibrium based on \eqref{FI_t} is also an equilibrium when the insider knows both $(\ta,\tv)$ initially in the sense $\FF_t^I:=\sigma(\tv, \ta, (W_u)_{u\in [0,t]})$.

The insider's objective is to maximize her expected profit subject to the constraint $\theta_T=\ta$:
\begin{align}\label{rebalancer_problem}
\sup_{ \theta \in \cA } \bE \left[
(\tv - P_T)\theta_T +\int_{[0,T]} \theta_{t-} dP_t   
\Big| \FF_0^I \right]= 
\rho\frac{\sigma_\tv}{\sigma_\ta } \ta^2- \inf_{ \theta \in \cA } \bE \left[
\int_{[0,T]} (\ta- \theta_{t-}) dP_t   
\Big| \FF_0^I \right],
\end{align} 
where the set $\cA$ is defined in Definition \ref{ad_def}  below. The first objective in \eqref{rebalancer_problem} is from Back (1992) and reflects that $\tv$ is the stock's liquidating dividends. The second objective in \eqref{rebalancer_problem} follows from  $\theta_T=\ta$ and $P_{0-}=0$ which give 
$$
P_T \theta_T = P_T\ta = (P_T-P_{0-})\ta = \ta\int_{[0,T]} dP_t.
$$
The joint normality of $(\ta,\tv)$ and $\FF_0^I = \sigma(\ta)$ give 
$$
\bE[\tv\theta_T|\FF_0^I] = \bE[\tv\ta| \sigma(\ta)] = 
\rho\frac{\sigma_\tv}{\sigma_\ta } \ta^2. 
$$

Based on the form of the insider's objective in \eqref{rebalancer_problem}, 
we conjecture that the insider's state variables are her remaining trading gab $\ta-\theta_{t}$ and the market makers' estimate of $\ta-\theta_{t}$ (denoted by $Q_t$). The  market makers prescribe a pricing rule $P_t$ (i.e., how prices depend on past and current aggregate holdings $(Y_s)_{s\in[0,t]}$) in terms of the following SDEs 
\begin{align}
dQ_t &= r(t)dY_t + s(t) Q_t dt \quad 
\textrm{for  $t\in[0,T)$}, \quad 
Q_{0-}=0, \label{q_dynamics}\\
dP_t &= \lambda(t)dY_t + \mu(t) Q_t dt \quad 
\textrm{for  $t\in[0,T)$}\label{p_dynamics}, \quad 
P_{0-}=0, \\
\Delta P_T &= \lambda(T) \big(\ta- \theta_{T-} - Q_{T-}\big), \label{p_terminal}
\end{align}
for deterministic functions of time $(r,s,\lambda,\mu)$ (these functions are to be determined in equilibrium). The insider uses $P_t$ when solving her optimization problem \eqref{rebalancer_problem}. As in Kyle (1985), equilibrium prices must be efficient in the sense that \eqref{p_rational} below holds. Additionally, in the equilibrium we construct in \thmref{main_thm} below, the process $Q_t$ from \eqref{q_dynamics} is the market makers' estimate of the insider's remaining trading gab $\ta-\theta_t$ (see \eqref{q_rational} below).

 As usual, to rule out doubling strategies and give a verification proof, we need to limit the insider's choice of possible  controls.

\begin{definition}[Admissibility]\label{ad_def} 
In addition to the initial condition $\theta_{0-}:=0$ and the terminal constraint $\theta_T=\ta$,
admissible holding processes $(\theta_t)_{t\in[0,T]}\in\cA$ satisfy: 

\begin{itemize}
\item[(i)] The holding process $\theta_t$ is a c\'adl\'ag semimartingale which is adapted to $\FF^I_t$ and square integrable $\bE\left [\int_0^T \theta_t^2 dt \right]<\infty$.

\item[ (ii)] For given continuous functions $r, s:[0,T)\to \bR$, a c\'adl\'ag solution $(Q_t)_{t\in[0,T)}$ of \eqref{q_dynamics} exists with an almost surely finite limit $Q_{T-}$ and  $\bE\left [\int_0^T Q_t^2 dt \right]<\infty$.

\item[(iii)] For given continuous functions $\lambda:[0,T]\to \bR$ and $\mu:[0,T)\to \bR$, a c\'adl\'ag solution $(P_t)_{t\in[0,T]}$ of \eqref{p_dynamics}-\eqref{p_terminal} exists such that  the stochastic integral 
$\int_{[0,t]} (\ta- \theta_{u-}) dP_u$,  $t\in[0,T]$, is a well-defined semimartingale with $\int_{[0,T]} (\ta- \theta_{t-}) dP_t$ integrable.
\begin{flushright}
$\square$
\end{flushright}

\end{itemize}

\end{definition}

The admissible set $\cA$ in Definition \ref{ad_def}  allows the insider to place block orders $\Delta \theta_t:= \theta_t - \theta_{t-}$ for all  $t\in[0,T]$. For example, the insider is allowed to place an initial buy order $\Delta \theta_0=1$, which produces the initial stock price 
$$
P_0 = P_{0-}+\Delta P_0 = \lambda(0)\Delta Y_0 = \lambda(0)\Delta \theta_0 = \lambda(0).
$$
However, we shall see in \thmref{main_thm} below that the insider optimally chooses to trade absolutely continuously at a $dt$ rate $\theta'_t$ for $t\in[0,T)$ and then places a block order at $t=T$ to satisfy her terminal restriction $\theta_T=\ta$ (i.e., $\Delta \theta_T = \ta- \theta_{T-}$). We note that while $Q_{T-}$ must exist for $\theta\in\mathcal{A}$, there is no $Q_T$ in  Definition \ref{ad_def}. 

\bigskip

\noindent{\emph{Market makers:} } These traders choose the pricing rule to clear the stock market. Market makers observe  the aggregate orders $dY_t$ from \eqref{Y}; hence, their filtration is defined as 
\begin{align}\label{sFM}
\FF_t^M:=\sigma\left((Y_u)_{u\in [0,t]}  \right),\quad t\in [0,T].
\end{align} 
The filtration $\FF_t^M$ is endogenous because $Y_t$ in \eqref{Y} depends on the insider's stock-holdings $\theta_t$. As in Kyle (1985) and Back (1992), in equilibrium, the stock price must be efficient in the sense
\begin{align}\label{p_rational}
P_t= \bE\left[\tv  |  \FF_t^M  \right], \quad t\in [0,T],
\end{align}
holds. The martingale condition in \eqref{p_rational} is a zero-profit requirement stemming from competition among risk-neutral market makers.

\medskip

\subsection{Equilibrium}
Similar to Kyle (1985), we consider only linear equilibria in this paper:

\begin{definition}[Kyle equilibrium]\label{equilibrium_def}
Continuous functions $\lambda:[0,T]\to \bR$ and $\mu,r,s,\beta,\alpha: [0,T)\to \bR$ constitute an equilibrium if:

\begin{itemize}
\item [(i)] For the pricing rule \eqref{p_dynamics}-\eqref{p_terminal} with $Q_t$ in \eqref{q_dynamics}, the stock-holding process
\begin{align}
d\theta_t &= \Big(\beta(t) (\ta -\theta_t-Q_t) + \alpha(t) Q_t\Big) dt  \quad 
\textrm{for $t\in [0,T)$}, \quad \theta_{0-}=0,
 \label{theta_dynamics}\\
\Delta \theta_T &= \ta- \theta_{T-},  \label{theta_terminal}
\end{align}
is in $\cA$ and maximizes the insider's objective \eqref{rebalancer_problem}.

\item[(ii)] The stock-holding process $\theta_t $ in \eqref{theta_dynamics}-\eqref{theta_terminal} with $Q_t$ in \eqref{q_dynamics} implies that $P_t$ in \eqref{p_dynamics}-\eqref{p_terminal} satisfies \eqref{p_rational}.
\end{itemize}
\begin{flushright}
$\square$
\end{flushright}
\end{definition}
\noindent While not a requirement in Definition \ref{equilibrium_def}, the process $Q_t$ given by \eqref{q_dynamics} satisfies $Q_t=\bE[\ta-\theta_t | \FF_t^M]$ in the equilibrium that we construct in Theorem \ref{main_thm} below.

\section{An autonomous system of ODEs}

Similar to Back (1992) and Back and Baruch (2004), we include a heuristic derivation of the value function and corresponding ODEs for the  constrained insider's optimization problem in \eqref{rebalancer_problem}. This section concludes with a result ensuring global existence of solutions to this system of coupled ODEs. In the next section, we use the solutions of the ODEs to prove existence of a Kyle equilibrium. 

To derive the HJB equation corresponding to the objective in \eqref{rebalancer_problem}, it suffices to consider stock-holding processes $\theta_t$ with $\theta_{0-}:= 0$ and dynamics 
\begin{align}
d\theta_t :=
\begin{cases}
&\theta'_tdt,\quad  t\in[0,T), \\
&\ta-\theta_{T-},\quad t=T,
\end{cases}
\end{align}
where $\theta'_t$ is an arbitrary order-rate process. Inserting the dynamics of \eqref{Y} into \eqref{q_dynamics}, we see that the state process
\begin{align}
X_t:= \ta- \theta_t - Q_t,\quad t\in[0,T), \quad X_{0-}:=\ta,\label{X_def}
\end{align}
has dynamics
\begin{equation}
\begin{split}\label{X_dynamics}
dX_t &=  -\theta'_tdt - r(t)\big( \theta'_t dt + \sigma_w dW_t\big) - s(t)Q_tdt,\quad t\in[0,T).
\end{split}
\end{equation}
We conjecture that the value function corresponding to the infimum in \eqref{rebalancer_problem} has the following quadratic structure:
\begin{align}
V(t,x,q):= Ix^2 + J(t)x q + K(t), \quad t\in[0,T],\quad x,q\in \mathbb{R},\label{V_form}
\end{align}
where $I>0$ is a constant and $(J,K)$ are deterministic functions of time. Because $I$ is a constant, the terminal condition for $V$ in \eqref{V_form} at time $t=T$ is not zero. This ``facelift'' feature was already present in the continuous-time model in Kyle (1985). However, unlike Kyle (1985), our equilibrium price-process $P_t$ does not converge to $\bE[\tv|\mathcal{F}_0^I] = \rho\frac{\sigma_\tv}{\sigma_\ta}\ta $ as $t\uparrow T$ (see \proref{prop_properties}(5)).

Formally, we get the HJB equation by equating the drift in the dynamics
\begin{align}\label{HJB1}
\begin{split}
(\ta-\theta_t)dP_t +dV(t,X_t,Q_t)
\end{split}
\end{align}
to zero, where $dP_t$ is from \eqref{p_dynamics}, $dQ_t$ is from \eqref{q_dynamics}, and $dX_t$ is from \eqref{X_dynamics}. It\^o's lemma shows that the drift of \eqref{HJB1} is an affine function of $\theta'_t$ given by a slope and an intercept. Consequently, for the drift of \eqref{HJB1} to be zero, the slope and intercept   must each separately be zero:
\begin{align*}
\begin{split}
0&=X_t\Big(\lambda (t)-2 I \big(r(t)+1\big)+J(t) r(t)\Big)+Q_t\Big(\lambda (t)-\big(r(t)+1\big)J(t) \Big),\\
0&=Q_t^2\big( \mu(t)-s(t)J(t)\big)+X_tQ_t\Big(s(t) \big(J(t)-2 I\big)+J'(t)+\mu (t)\Big)+K'(t)+\sigma_w^2\big(I-J(t)\big) r(t)^2.
\end{split}
\end{align*}
Equating coefficients for $X_t$ and $Q_t$ in the first equality and coefficients for $Q_t^2$, $X_tQ_t$, and deterministic terms in the second equality to zero gives the requirements
\begin{equation}\label{HJB3}
\begin{split}
0&=\lambda(t)-2I(r(t)+1)+J(t) r(t),\\
0&=\lambda(t)-(r(t)+1)J(t),\\
0&=\mu(t)-s(t)J(t),\\
0&=s(t) \big(J(t)-2 I\big)+J'(t)+\mu (t),\\
0&=K'(t)+\sigma_w^2 \big(I-J(t)\big)r(t)^2.
\end{split}
\end{equation}

Our equilibrium existence proof in the next section is based on an autonomous two-dimensional coupled system of nonlinear ODEs. To heuristically derive this ODE system, the Kalman-Bucy result from Gaussian filtering theory (recalled in Appendix \ref{app:KB}) shows that when $\theta_t$ is as in \eqref{theta_dynamics}, the property
\begin{align}
Q_t=\bE[\ta-\theta_t | \FF_t^M],\quad t\in[0,T),  \label{q_rational}
\end{align}
and \eqref{p_rational} impose the relations 
\begin{align}
\lambda(t)&=\frac{\beta(t) \Sigma_2(t)}{\sigma_w^2}, \quad \mu(t) = -\alpha(t) \lambda(t),\label{lm_exp}\\
r(t) &= \frac{\beta(t) \Sigma_1(t)}{\sigma_w^2}, \quad s(t) =- \alpha(t)(1+r(t)), \label{rs_exp}
\end{align}
where 
$\Sigma_1(t)$ and $\Sigma_2(t)$ are defined as solutions to the ODEs
\begin{align}
\Sigma_1'(t) &= -\sigma_w^2(r(t)^2 + 2r(t)),\quad t\in [0,T),\quad \Sigma_1(0)= \sigma_\ta^2, \label{S1_ode}\\
\Sigma_2'(t) &= -\sigma_w^2 \big(1+r(t)\big) \lambda(t), \quad t\in [0,T), \quad \Sigma_2(0)=\rho \sigma_\ta \sigma_\tv.\label{S3_ode}
\end{align}
 In \eqref{Sigmas} below, we shall see that $\Sigma_1(t)$ and $\Sigma_2(t)$ are the remaining variance and covariance functions produced by Kalman-Bucy filtering theory (Appendix \ref{app:KB} recalls the result from Kalman-Bucy filtering theory that we use). The relation between $\mu(t)$ and $s(t)$ in \eqref{lm_exp}-\eqref{rs_exp} produces redundancy in the second and third equation in \eqref{HJB3}.

By solving the first two equations in \eqref{HJB3} we get $I=\frac{(1+2r(t))\lambda(t)}{2(1+r(t))^2}$ and so
\begin{align}\label{HJB4}
\begin{split}
I\big(1+r(t)\big)^2&=\tfrac12\lambda(t)\big(1+2r(t)\big)\\
&=\tfrac12\tfrac{r(t)\Sigma_2(t)}{\Sigma_1(t)}\big(1+2r(t)\big),
\end{split}
\end{align}
where the second equality is due to \eqref{lm_exp}-\eqref{rs_exp}. By using \eqref{S1_ode}-\eqref{S3_ode} when computing time derivatives in \eqref{HJB4}, we produce the ODE
\begin{align}\label{rkey1}
 r'(t) =- \frac{\sigma_w^2r(t)^2 \big(1+r(t)\big)(1+2r(t))}{\big(1+3r(t)\big)\Sigma_1(t)},\quad t\in[0,T).
\end{align}
The ODEs \eqref{S1_ode} and \eqref{rkey1} form an autonomous system. The correct initial condition $r(0)=r_0$ for \eqref{rkey1} is such that  $\Sigma_1'(t)$ in \eqref{S1_ode} satisfies $\Sigma_1(T-)=0$. The intuition behind $\Sigma_1(T-)=0$ comes from the representation in \eqref{q_rational}, which gives the market makers' conditional variance of $\ta - \theta_t$ as
\begin{align}\label{HJB44}
\begin{split}
\bE\big[\big( \ta - \theta_t -\bE[\ta-\theta_t | \FF_t^M]\big)^2| \FF_t^M \big] & = \bE\big[\big( \ta - \theta_t -Q_t\big)^2| \FF_t^M \big]\\
&=\Sigma_1(t).
\end{split}
\end{align}
The last equality in \eqref{HJB44} is based on the Gaussian structure, which gives us that $\ta - \theta_t -Q_t$ is independent of $\FF_t^M$ when \eqref{q_rational} holds and $\Sigma_1(t)$ being the remaining variance function. The terminal condition $\Sigma_1(T-)=0$ can be interpreted as the market makers expect $\theta_t$ to converge to $\ta$ as $t\uparrow T$. As we shall see in \proref{prop_properties}(4), even though $Q_{T-} \neq 0$ in equilibrium, $\Sigma_1(T-)=0$ is the correct boundary condition because the insider's optimal terminal block order $\Delta\theta_T =\ta-\theta_{T-}\neq 0$ can be predicted by the market makers.

While the above discussion is purely heuristic, the next  result is rigorous and guarantees global existence to above coupled ODE system.

\begin{lemma}\label{ode_existence_lemma}
There exists a constant $r_0\in (0,\infty)$ such that the coupled ODEs \eqref{S1_ode} and \eqref{rkey1}
with initial conditions 
\begin{align}
r(0)=r_0 \quad \textrm{and} \quad \Sigma_1(0)=\sigma_\ta^2, \label{r0_boundary_cond}
\end{align}
 have global solutions in $\mathcal{C}^1([0,T])$ that satisfy
 \begin{align}\label{BC22}
\Sigma_1(T):= \Sigma_1(T-)=0,\quad r(T):=r(T-)=0, \quad r(t),\Sigma_1(t)>0\quad  \text{ for } t\in [0,T).
 \end{align}
\end{lemma}

\begin{proof}
Let $F:[0,\infty)  \to [\pi-3,2\pi)$ be defined as
\begin{align}
F(x):=4 \tan^{-1}(\sqrt{1+2x})- \frac{\sqrt{1+2x} \, (3+4x)}{(1+x)^2}. \label{F_def}
\end{align}
Then, $F:[0,\infty)  \to [\pi-3,2\pi)$  is strictly increasing and bijective because
$F'(x)=\frac{(1+3x)\sqrt{1+2x}}{(1+x)^3}>0$ for $x\geq 0$. Hence, its inverse function $F^{-1}:  [\pi-3,2\pi) \to [0,\infty)$ is well-defined. 

Let $G,\tau :(0,\infty)\to [0,\infty)$ be defined as
\begin{align*}
G(x)&:= \frac{x^2(1+2x)^{\frac{3}{2}}}{(1+x)^2},\\
\tau(x)&:=\frac{\sigma_\ta^2(F(x)-F(0))}{\sigma_w^2 G(x)}. 
\end{align*}
For $r_0 \in  (0,\infty)$, the functions
\begin{equation}
\begin{split}\label{ode_explicit_solution}
\begin{cases}
r(t):=F^{-1}\left(  F(r_0)- \tfrac{\sigma_w^2}{\sigma_\ta^2}G(r_0) t \right), \\
\Sigma_1(t):=\frac{\sigma_\ta^2 G(r(t))}{G(r_0)},
\end{cases}
\end{split}
\end{equation} 
are well-defined for $t\in [0,\tau(r_0)]$ with $r(0) =r_0$. Direct computations show that \eqref{ode_explicit_solution} satisfies \eqref{rkey1} and \eqref{r0_boundary_cond}  for $t\in [0,\tau(r_0))$. Furthermore, the above functions $r(t)$ and $\Sigma_1(t)$ satisfy
\begin{align*}
&r(\tau(r_0))=0,\quad \Sigma_1(\tau(r_0))=0,\\
&r(t),\Sigma_1(t)>0 \quad \textrm{for} \quad t\in [0,\tau(r_0)).
\end{align*}

To complete the proof, it remains to be proven that there exists $r_0\in (0,\infty)$ such that $\tau(r_0)=T$.
L'Hopital's rule gives
\begin{equation}
\begin{split}\label{tau_limits}
&\lim_{x \downarrow 0} \tau(x)= \infty, \quad \lim_{x\to \infty} \tau(x) = 0.
\end{split}
\end{equation} 
These limits and the continuity of $\tau(x)$ for $x\in(0,\infty)$ imply that there exists $r_0\in(0,\infty)$ such that $\tau(r_0)=T$.
\end{proof}

\section{Existence of a Kyle equilibrium}

The next result is our main contribution and the theorem ensures the existence of a global Kyle equilibrium in the sense of Definition \ref{equilibrium_def}.

\begin{theorem}\label{main_thm} Let $\rho \in (0,1]$ and let $r(t)$ and $\Sigma_1(t)$ be as in  \lemref{ode_existence_lemma}. Define the constant $I:=\frac{ \rho\sigma_\tv}{\sigma_\ta} \frac{ r_0(1+2r_0)}{2(1+r_0)^2}>0$ and  the functions
\begin{equation}\label{equilibrium_explicit}
\begin{split}
\lambda(t)&:=2I\frac{\big(1+r(t)\big)^2}{1+2r(t)} ,\quad t\in[0,T],\\
\mu(t) &:=- \frac{\rho \sigma_w^2 \sigma_\tv r_0^3(1+2r_0)^{\frac{5}{2}}}{\sigma_\ta^3 (1+r_0)^4}
\frac{(1+r(t))^4}{(1+3r(t))(1+2r(t))^{\frac{5}{2}}}, \quad t\in[0,T],\\
\beta(t)&:=\frac{\sigma_w^2 r_0^2(1+2r_0)^{\frac{3}{2}}}{\sigma_\ta^2 (1+r_0)^2}
 \frac{(1+r(t))^2}{r(t)(1+2r(t))^{\frac{3}{2}}}, \quad t\in[0,T),\\
s(t)&:=-\frac{\sigma_w^2 r_0^2(1+2r_0)^{\frac{3}{2}}}{\sigma_\ta^2 (1+r_0)^2}
 \frac{(1+r(t))^3}{(1+3r(t))(1+2r(t))^{\frac{3}{2}}}, \quad t\in[0,T],\\
\alpha(t)&:=\frac{\sigma_w^2 r_0^2(1+2r_0)^{\frac{3}{2}}}{\sigma_\ta^2 (1+r_0)^2}
 \frac{(1+r(t))^2}{(1+3r(t))(1+2r(t))^{\frac{3}{2}}}, \quad t\in[0,T].
\end{split}
\end{equation}
Then, the functions $\lambda,\mu,r,s,\beta,\alpha$ constitute a Kyle equilibrium, where the process $Q_t$ additionally satisfies \eqref{q_rational}.

\end{theorem}

\begin{proof}

\noindent  We start by defining the function  
\begin{align}
\begin{split}
\Sigma_2(t):=&\frac{\lambda(t) \Sigma_1(t)}{r(t)} \\
=&\frac{\rho \sigma_\ta \sigma_\tv}{r_0 \sqrt{1+2r_0}} r(t)\sqrt{1+2r(t)},\label{Sigma2}
\end{split}
\end{align}
for $t\in[0,T]$. Since $r(t)$ is continuous on $t\in [0,T]$ and $r(T)=0$, $\Sigma_2(t)$ is continuous on $t\in[0,T]$ with 
$\Sigma_2(T)=0$.

We divide the proof into three steps:

\noindent{\bf Step 1/3:} In this step, we show that the value function corresponding to \eqref{rebalancer_problem} is greater than or equal to $V(t,x,q)$ in \eqref{V_form} for the coefficient functions $J,K:[0,T]\to \bR$ defined by
\begin{equation}
\begin{split}\label{IJK_explicit}
J(t)&:=\frac{\lambda(t)}{1+r(t)},\\
K(t)&:=\sigma_w^2 \int_t^T \big(I-J(u)\big)r(u)^2 du.
\end{split}
\end{equation}
We let $X_t$ be as in \eqref{X_def}. Then, for $\theta\in \cA$ arbitrary, we have
\begin{equation}
\begin{split}\label{quad_variation}
d[X,X]^c_t &= \big(1+r(t)\big)^2 d[\theta,\theta]^c_t + \sigma_w^2 r(t)^2 dt + 2\sigma_w\big(1+r(t)\big)r(t) d[\theta,W]^c_t,\\
 d[X,Q]^c_t &= -r(t)\big(1+r(t)\big) d[\theta,\theta]^c_t - \sigma_w^2 r(t)^2 dt - \sigma_wr(t)\big(2r(t)+1\big) d[\theta,W]^c_t,
\end{split}
\end{equation}
where $[\cdot,\cdot]^c$ denotes the continuous part of the quadratic covariation process $[\cdot,\cdot]$; see, e.g., p.70 in Protter (2005). For $t\in [0,T)$, Ito's formula gives
\begin{align}
&\int_{[0,t]} (\ta-\theta_{u-})dP_u + V(t,X_t,Q_t)   \nonumber\\
&= \int_{[0,t]} (\ta-\theta_{u-})\lambda(u)\big(d\theta_u + \sigma_wdW_u-  \alpha(u) Q_u du\big) \nonumber \\
&\quad +I\ta^2+  K(0) +\int_{[0,t]} \big(J'(u) X_{u} Q_u + K'(u)\big) du \nonumber \\
&\quad + \int_{[0,t]} \Big( \big( 2I X_{u-}+J(u)Q_{u-} \big)dX_u + J(u)X_{u-} dQ_u +I d[X,X]^c_u + J(u) d[X,Q]^c_u \Big) \nonumber\\
&\quad + \sum_{0\leq u\leq t} \Big( \Delta V(u,X_u,Q_u) - \big( 2I X_{u-}+J(u)Q_{u-})\big)\Delta X_u - J(u)X_{u-} \Delta Q_u \Big)  \nonumber\\
& = I\ta^2+K(0)+ \int_0^t \Big(  \big(\lambda(u)-2r(u)I+r(u)J(u) \big)X_u + \big(\lambda(u)-r(u)J(u)\big)Q_u  \Big)\sigma_w dW_u  \nonumber \\
&\quad  +\tfrac12\int_0^t\lambda(u) d[\theta,\theta]^c_u+\tfrac12\sum_{0\leq u\leq t} \lambda(u)(\Delta \theta_u)^2, \label{simplified}
\end{align}
where we used $\ta-\theta_{u-}=X_{u-}+Q_{u-}$, \eqref{equilibrium_explicit}, \eqref{IJK_explicit}, and \eqref{quad_variation} to obtain the second equality. In \eqref{simplified}, the stochastic integral with respect to $dW_u$ is a martingale on $t\in [0,T]$ because of the square integrability requirement in Definition \ref{ad_def}. Passing $t\uparrow T$ in \eqref{simplified} produces the limit
\begin{align*}
&\int_{[0,T)} (\ta-\theta_{u-})dP_u + V(T,X_{T-},Q_{T-}) \\
&=\lim_{t \uparrow T} \left(  \int_{[0,t]} (\ta-\theta_{u-})dP_u + V(t,X_t,Q_t)  \right)\\
&=I\ta^2+K(0)+ \int_0^T \Big(  \big(\lambda(u)-2r(u)I+r(u)J(u) \big)X_u + \big(\lambda(u)-r(u)J(u)\big)Q_u  \Big)\sigma_w dW_u  \nonumber \\
&\quad  +\tfrac12\int_0^T\lambda(u) d[\theta,\theta]^c_u+\tfrac12\sum_{0\leq u< T} \lambda(u)(\Delta \theta_u)^2.
\end{align*}
Because the last two terms are positive and the stochastic integral with respect to $dW_u$ is a martingale for any $\theta\in \cA$ we have
\begin{align}
\bE\left[\int_{[0,T)} (\ta-\theta_{u-})dP_u + V(T,X_{T-},Q_{T-}) \Big| \FF_0^I \right] \geq I\ta^2+ K(0). \label{value_ineq}
\end{align}
By using $r(T-)=0$,  $J(t)$ in \eqref{IJK_explicit}, and $\lambda(t)$ in  \eqref{equilibrium_explicit} we obtain\footnote{The functions $\lambda(t)$ and $J(t)$ are continuous functions on $t\in[0,T]$, so we have $\lambda(T-)=\lambda(T)$ and $J(T-)=J(T)$. Except for $\beta$, the other functions $(r, \Sigma_1,\Sigma_2, \mu, s, \alpha)$ are also continuous on $t\in[0,T]$.}
\begin{align}
J(T-)=\lambda(T)=2I>0. \label{lIJ}  
\end{align}
Combing \eqref{lIJ} with \eqref{p_terminal} produces the inequality
\begin{align}
(\ta-\theta_{T-})\Delta P_T&=(\ta-\theta_{T-})\lambda(T) (\ta-\theta_{T-}-Q_{T-})\nonumber\\
&=2IX_{T-}^2+ J(T-)X_{T-}Q_{T-} \nonumber\\
&\geq V(T,X_{T-},Q_{T-}). \label{jump_ineq}
\end{align}
We combine \eqref{value_ineq} and \eqref{jump_ineq} to conclude that 
\begin{align*}
\inf_{\theta\in \cA}\bE\left[ \int_{[0,T]} (\ta-\theta_{u-})dP_u  \Big| \FF_0^I \right]&= \inf_{\theta\in \cA}\bE\left[ \int_{[0,T)} (\ta-\theta_{u-})dP_u + (\ta-\theta_{T-})\Delta P_T \Big| \FF_0^I \right]\\
&\geq \inf_{\theta\in \cA} \bE\left[ \int_{[0,T)} (\ta-\theta_{u-})dP_u +V(T,X_{T-},Q_{T-})\Big| \FF_0^I  \right] \\
&\geq I\ta^2+K(0).
\end{align*}

\noindent{\bf Step 2/3:} We show that $V(t,x,q)$ in \eqref{V_form} is indeed the 
 the value function corresponding to \eqref{rebalancer_problem}, and
 the stock-holding process $\theta_t$ in \eqref{theta_dynamics}-\eqref{theta_terminal} is admissible and optimal.  In this step, we let $(\theta_t,P_t,Q_t)$ be the solution of the SDE system \eqref{q_dynamics}, \eqref{p_dynamics} and \eqref{theta_dynamics} on $t\in [0,T)$. The existence of the solution is ensured by the continuity of $\lambda,\mu,r,s,\beta,\alpha:[0,T) \to \bR$.
 From the Kalman-Bucy result in \lemref{filter_lem}, the representations in \eqref{p_rational} and \eqref{q_rational} hold for $t\in [0,T)$ and we have
\begin{equation}
\begin{split}\label{Sigmas}
\Sigma_1(t) &= \bE[(\ta - \theta_t -Q_t)^2],\quad  t\in [0,T),\\
\Sigma_2(t) &= \bE[(\tv-P_t)(\ta - \theta_t -Q_t)],\quad  t\in [0,T).
\end{split}
\end{equation}
We observe that \eqref{q_rational} and $Q_t\in \FF_t^M$ imply
\begin{align}
\bE\left[ Q_t(\ta-\theta_t-Q_t) \right]=\bE\left[ Q_t\bE\left[\ta-\theta_t-Q_t | \FF_t^M\right] \right]=0,\quad t\in[0,T).\label{QX0}
\end{align}
Then, for $t\in [0,T)$, the dynamics $dQ_t$ in \eqref{q_dynamics} and It\^o's lemma produce 
\begin{align}\label{S3a}
\begin{split}
\Sigma_3(t):&= \bE[Q_t^2]\\
&=\bE\left[ \int_0^t 2Q_u dQ_u + \sigma_w^2 r(u)^2 du   \right]\\
&=\int_0^t  \Big(- 2\alpha(u)\Sigma_3(u)+ \sigma_w^2 r(u)^2 \Big) du,
\end{split}
\end{align}
where the last equality is due to \eqref{QX0} and \eqref{lm_exp}-\eqref{rs_exp}. By computing the time derivative, we get the ODE
\begin{align}\label{S3b}
\begin{split}
\Sigma_3'(t)&=-2\alpha(t)\Sigma_3(t) + \sigma_w^2 r(t)^2, \quad \Sigma_3(0)=0.
\end{split}
\end{align}
Similarly, the function $\Sigma_4(t):=\bE[(\tv-P_t)^2]$, $t\in[0,T)$, satisfies the ODE
\begin{align}
\Sigma_4'(t)&=-\sigma_w^2 \lambda(t)^2, \quad \Sigma_4(0)=\sigma_\tv^2, \label{Sigma4_ode}
\end{align}
where we have used $dP_t$ in \eqref{p_dynamics}, the representation of $\Sigma_2$ in \eqref{Sigmas}, and  \eqref{lm_exp}.

The boundedness of $r,\alpha,\lambda$ on $t\in [0,T]$ and the above ODEs  give
\begin{align}
\sup_{t\in [0,T)} \Sigma_i(t)<\infty,\quad \textrm{for}\quad i \in\{1,2,3,4\}.\label{Sigma_bdd}
\end{align}
The explicit expressions of $\beta(t)$ in \eqref{equilibrium_explicit} and $\Sigma_1(t)$ in \eqref{ode_explicit_solution} produce
\begin{align}
\sup_{t\in [0,T)} \beta(t)^2\Sigma_1(t)=
\sup_{t\in [0,T)}
\frac{\sigma_w^4 r_0^2(1+2r_0)^{\frac{3}{2}}}{\sigma_\ta^2 (1+r_0)^2}
 \frac{(1+r(t))^2}{(1+2r(t))^{\frac{3}{2}}}
 <\infty. \label{betaSigma_bdd}
\end{align}
Based on these bounds, we can verify the conditions in Definition \ref{ad_def}:

(i) The following observation and \eqref{Sigma_bdd}-\eqref{betaSigma_bdd} produce $\bE\left[\int_0^T \theta_t^2 dt\right]<\infty$:
\begin{align}
\bE[\theta_t^2]&=\bE\left[ \left( \int_0^t \big(\beta(u)(\ta-\theta_u-Q_u)+\alpha(u)Q_u\big)  du  \right)^2  \right]\\
&=C \int_0^t  \Big( \beta(u)^2 \Sigma_1(u)+\alpha(u)^2 \Sigma_3(u)\Big) du, \quad t\in[0,T),
\end{align}
for a constant $C$ independent of $t$.

(ii) The expectation $\bE\left[\int_0^T Q_t^2 dt\right]=\int_0^T\Sigma_3(t)dt$ is finite by \eqref{Sigma_bdd}. It\^o's lemma gives the dynamics of $Z_t:=e^{\int_0^t \alpha(u)du} Q_t$, $t\in [0,T)$, as
$$
dZ_t = e^{\int_0^t \alpha(u)du} r(t)\big( \beta(t)(\ta -\theta_t-Q_t)dt + \sigma_wdW_t \big), 
$$
where we have used the relation \eqref{rs_exp}. The representation \eqref{q_rational} ensures that the process with dynamics  $\beta(t)(\ta -\theta_t-Q_t)dt + \sigma_wdW_t$ is the market makers' innovations process. Therefore, $Z_t$ is a martingale with respect to  $\FF_t^M$. Furthermore, \eqref{betaSigma_bdd} implies that $Z_t$ is a square integrable martingale uniformly bounded in $\mathcal{L}_2(\PP)$:
\begin{align*}
\sup_{t\in [0,T)}\bE[ Z_t^2] &= \sup_{t\in [0,T)}\bE\left[ \left(  \int_0^t  e^{\int_0^u \alpha(v)dv}r(u)\beta(u)(\ta -\theta_u - Q_u)du + e^{\int_0^u \alpha(v)dv} r(u)  dW_u \right)^2     \right]<\infty,
\end{align*}
where we have used the boundedness of $r(t)$ and $\alpha(t)$ on $t\in[0,T]$ and \eqref{Sigma_bdd}-\eqref{betaSigma_bdd}. Therefore, a finite limit $\lim_{t\uparrow T}Z_{t}$ exists, which implies the existence of $Q_{T-}$ (here we also use the continuity of $\alpha$ on $t\in[0,T]$).

(iii) Similarly to (ii) above, the process $(P_t)_{t\in [0,T)}$ is a square integrable martingale uniformly bounded in $\mathcal{L}_2(\PP)$. This implies the existence of $P_{T-}$.

Finally, we show $\theta_t$'s optimality. Because $\Delta \theta_t=0$ and $[\theta,\theta]_t=0$ for $t\in [0,T)$, the inequality \eqref{value_ineq} becomes an equality. The representation of $\Sigma_1$ in \eqref{Sigmas} and the boundary condition $\Sigma_1(T-)=0$ in \eqref{BC22} give
\begin{align}
\begin{split}
0 &=\lim_{t\uparrow T}   \bE[(\ta-\theta_{t}-Q_{t})^2]\\
&\ge \bE[(\ta-\theta_{T-}-Q_{T-})^2],
\end{split}
\end{align}
where the inequality comes from Fatou's lemma. Therefore, 
\begin{align}
X_{T-}=\ta-\theta_{T-}-Q_{T-}=0, \quad \textrm{almost surely,} \label{predictable_jump}
\end{align}
and the inequality \eqref{jump_ineq} becomes an equality.
All in all, $\theta_t$ in \eqref{theta_dynamics}-\eqref{theta_terminal} is optimal and satisfies
\begin{align*}
\bE\left[ \int_{[0,T]} (\ta-\theta_{u-})dP_u \Big| \FF_0^I \right]=I\ta^2+ K(0).
\end{align*} 

\medskip

\noindent{\bf Step 3/3:} It remains to verify the property in \defref{equilibrium_def} (ii). For $t\in [0,T)$, this follows from the Kalman-Bucy result in \lemref{filter_lem} in Appendix \ref{app:KB}. Therefore, we  only need to verify that  \eqref{p_rational} holds for $t=T$. To this end, we observe 
\begin{align}
\bE\left[\tv | \FF_{T}^M\right]
&=\bE\left[\tv | \FF_{T-}^M\right]\nonumber\\
&=P_{T-}=P_T, \label{P_no_jump}
\end{align}
where the first equality is due to \eqref{predictable_jump} and $Q_{T-}\in \FF_{T-}^M$, and the last equality is due to \eqref{p_terminal} and \eqref{predictable_jump}.
\end{proof}

Instead of Gaussian dividends $\tv$, Back (1992) and Cho (2003) consider $h(\tv)$  for a strictly increasing function $h:\mathbb{R}\to \mathbb{R}$ as the stock's terminal dividends. In a similar spirit, instead of the terminal target $\ta$, it would be interesting to consider targets of the form $h(\ta)$ for a strictly increasing function $h$.  Unfortunately, our setting does not immediately extend this case and we leave it for future research to identify an extended set of sufficient state processes for such non-Gaussian targets.

\section{Properties of equilibrium}

This section discusses new and desirable model features produced by the insider's trading constraint $\theta_T=\ta$. 
First, our equilibrium aggregate orders are  positively autocorrelated whereas in Kyle (1985) this process is a Brownian motion (hence, zero autocorrelation).   \c{C}et\.in and Danilova (2016) consider market makers with exponential utilities and produce mean reversion in the equilibrium aggregate holding process.  

Second,  our equilibrium price-impact function $\lambda(t)$ is decreasing over the trading interval $t\in[0,T]$. Barardehi and Bernhardt (2021) show that decreasing price-impact functions are empirically desirable, and while the discrete-time model in Kyle (1985) also has a decreasing price-impact function, price impact is constant in the continuous-time model in Kyle (1985).  In Kyle models with an exponentially distributed time horizon, Back and Baruch (2004), Caldentey and Stacchetti (2010), and \c{C}et\.in (2018) show that $\lambda$ is a supermartingale. Baruch (2002) and Cho (2003) show that an insider with exponential utility produces a deterministic and decreasing $\lambda$. For nonlinear pricing rules and market makers with exponential utilities, \c{C}et\.in and Danilova (2016) can produce both super- and submartingales for $\lambda$.

Third, in Kyle (1985), the insider's optimal order rate process $x'_t$ satisfies $\bE[x'_t | \tv]=\frac{\sigma_w}{\sigma_\tv}\tv$ for all $t\in[0,T]$. However,  empirically, there is more trading in the morning and afternoon relative to the middle of the day when $[0,T]$ is one trading day (see, e.g., Choi, Larsen, and Seppi (2019)). In expectation, our model produces such optimal $U$ shaped trading behavior for the insider. 

Fourth, our constrained insider optimally places a terminal block order $\Delta \theta_T=\ta - \theta_{T-}\neq 0$ whereas in Back (1992) such orders are suboptimal. However, the terminal block order is predictable by the market makers (i.e., $\ta - \theta_{T-}$ is $\FF_{T-}^M$ measurable) and does not induce a jump of the stock price (i.e., $\Delta p_T =0$). 

Fifth, unlike Kyle (1985) and Back (1992), the equilibrium stock price at the terminal time is different from the insider's initial expectation of $\tv$ (i.e., $P_T \neq \bE[\tv|\mathcal{F}_0^I] )$. This means that our constrained insider still has some unrevealed private information at the terminal time. Barger and Donnelly (2021) show that transaction costs can also produce this property.

The next proposition states the aforementioned characteristics of our equilibrium.

\begin{proposition}\label{prop_properties} In the setting of Theorem \ref{main_thm}, we have:
\begin{enumerate}
\item The scaled autocorrelation of equilibrium aggregate holdings  is positive
\begin{align}\label{scaledauto}
\lim_{h \downarrow0} \frac{1}{h}\frac{\bE[(Y_t-Y_{t-h})(Y_{t+h}-Y_t)] }{\sqrt{\mathbb{V}[Y_t-Y_{t-h}]\mathbb{V}[Y_{t+h}-Y_{t}]}}
&=\alpha(t)\left( \frac{\alpha(t) \Sigma_3(t)}{\sigma_w^2} + r(t)  \right)>0,\quad t\in(0,T),
\end{align}
where the positive function $\Sigma_3(t)$ is defined in \eqref{S3a}. Furthermore, in equilibrium, the process $Q_t$ in \eqref{q_dynamics} is mean reverting.

\item The price-impact function $\lambda(t)$ is decreasing for $t\in[0,T]$.

\item For $\ta \neq 0$, the mapping $[0,T) \ni t \mapsto \frac{1}{\ta}\bE[\theta'_t|\FF_0^I]$ is $U$ shaped where $\theta'_t$ is  the insider's equilibrium order-rate process\footnote{It turns out that $\bE[\theta'_t|\FF_0^I]$ is linear in $\ta$; hence, the ratio $\bE[\theta'_t|\FF_0^I]/\ta$ does not depend on $\ta$.} 
\begin{align}\label{thetaprime}
\theta'_t := \beta(t)(\ta-\theta_t)+ \big(\alpha(t) -\beta(t) \big)Q_t,\quad t\in[0,T),
\end{align}
and the terminal block order satisfies $0<\frac{1}{\ta}\bE[\Delta \theta_T | \FF_0^I]<1$.
\item $\Delta\theta_T = Q_{T-}\neq 0$ almost surely, $Q_{T-}\in \FF_{T-}^M$, and $\Delta P_T =0$.
\item $P_T \neq \bE[\tv|\mathcal{F}_0^I] $ almost surely.
\end{enumerate}
\end{proposition}

\begin{proof} To simplify, we give the proof for $\sigma_w:=1$.  

(1): Based on $dQ_t$ in \eqref{q_dynamics}, the dynamics of $\theta'_t$ in \eqref{thetaprime} have the form
\begin{align}
d\theta'_t &= A_t dt +  \big(\alpha(t) -\beta(t) \big) r(t) dW_t,\quad t\in[0,T),
\end{align}
for some integrable process $A_t$. For $h>0$, we have
\begin{align*}
\bE[(Y_t-Y_{t-h})(Y_{t+h}-Y_t)] 
&=\bE\left[ \left( \int_{t-h}^t\theta'_s ds + W_t-W_{t-h} \right) \left( \int_{t}^{t+h}\theta'_s ds + W_{t+h}-W_{t} \right)   \right]\\
&=\bE\left[ \left( \int_{t-h}^t\theta'_s ds + W_t-W_{t-h} \right)  \int_{t}^{t+h}\theta'_s ds     \right]\\
&=\bE\left[ \int_{t-h}^t\theta'_s ds  \int_{t}^{t+h}\theta'_s ds \right]+ \int_t^{t+h} \bE[(W_t-W_{t-h})\theta'_s] ds.
\end{align*}
The first term above can be approximated as $h^2 \bE[(\theta_t')^2]$ for $h>0$ close to $0$. For the second term, we let $s\in [t,t+h]$ and compute
\begin{align*}
\bE[(W_t-W_{t-h})\theta'_s]
&=\bE\left[ (W_t-W_{t-h})\left(\theta_0'+ \int_0^s \Big(A_u du + \big(\alpha(u)-\beta(u)\big)r(u) dW_u \Big) \right)   \right]\\
&=\bE\left[ (W_t-W_{t-h}) \int_{t-h}^s \Big(A_u du + \big(\alpha(u)-\beta(u)\big)r(u)dW_u \Big)    \right]\\
&=\int_{t-h}^t \big(\alpha(u)-\beta(u)\big)r(u) du + O(h^{3/2}),
\end{align*}
where we used the following observations:
\begin{align*}
&\bE\left[ (W_t-W_{t-h})\left( \int_{t-h}^t \Big(A_u du + \big(\alpha(u)-\beta(u)\big)r(u)dW_u \Big) \right)   \right] \\
& =\bE\left[ (W_t-W_{t-h}) \int_{t-h}^t A_u du \right] + \int_{t-h}^t \big(\alpha(u)-\beta(u)\big)r(u) du \\\
&=  \int_{t-h}^t  \big(\alpha(u)-\beta(u)\big)r(u) du+O(h^{3/2}), \\
&\bE\left[ (W_t-W_{t-h})\left( \int_{t}^s A_u du +  \big(\alpha(u)-\beta(u)\big)r(u)  dW_u  \right)   \right]\\
&=\bE\left[ (W_t-W_{t-h}) \int_{t}^s A_u du\right]\\
&=O(h^{3/2}).
\end{align*}
Therefore, we obtain the scaled autocorrelation:
\begin{align*}
\lim_{h\downarrow 0} \frac{1}{h}\frac{\bE[(Y_t-Y_{t-h})(Y_{t+h}-Y_t)] }{\sqrt{\mathbb{V}[Y_t-Y_{t-h}]\mathbb{V}[Y_{t+h}-Y_{t}]}}
&=\lim_{h\downarrow 0} \frac{\bE[(Y_t-Y_{t-h})(Y_{t+h}-Y_t)] }{h^2}\\
&= \bE[(\theta_t')^2] + \big(\alpha(t)-\beta(t)\big)r(t)\\
&=\beta(t)^2 \Sigma_1(t)+ \alpha(t)^2 \Sigma_3(t) +\big(\alpha(t)-\beta(t)\big)r(t)\\
&=\alpha(t)\left(\alpha(t) \Sigma_3(t)+ r(t)  \right)>0,
\end{align*}
where the third equality is due to \eqref{Sigmas} and \eqref{QX0}, and the last equality is due to $r(t)=\beta(t) \Sigma_1(t)$ from \eqref{rs_exp}.

From the proof of Theorem \ref{main_thm} we know that the market makers' innovations process (i.e.,  a Brownian motion generating the filtration \eqref{sFM}) is given by
$$
dW_t^Y := \beta(t)(\tilde{a} -\theta_t-Q_t)dt +dW_t,\quad W_0^Y:=0.
$$
The Brownian motion $W^Y_t$ allows us to rewrite the dynamics of $Q_t$ from \eqref{q_dynamics} as
\begin{align}
dQ_t &=  -\alpha(t) Q_t dt+ r(t) dW_t^Y.\label{eqQ}
\end{align}
Because $\alpha(t)$ from \eqref{equilibrium_explicit} is positive, \eqref{eqQ} shows that $Q_t$ is mean reverting.

(2): Using the ODE \eqref{rkey1} and the expressions of $\Sigma_1(t)$ and $\lambda(t)$ in \eqref{ode_explicit_solution} and \eqref{equilibrium_explicit}, we obtain
\begin{align*}
\lambda'(t)= -\frac{2\rho \sigma_w^2 \sigma_\tv r_0^3 (1+2r_0)^{\frac{5}{2}}}{\sigma_\ta^3 (1+r_0)^4} \frac{r(t) (1+r(t))^4}{(1+3r(t))(1+2r(t))^{\frac{5}{2}}}<0 \quad \textrm{for} \quad t\in (0,T).
\end{align*}

(3): Let $f,g: [0,T)\to \bR$ be defined as
\begin{align*}
f(t):=\frac{\bE[\theta_t | \FF_0^I]}{\ta}, \quad g(t):=\frac{\bE[Q_t | \FF_0^I]}{\ta}. 
\end{align*}
The SDEs \eqref{q_dynamics} and \eqref{theta_dynamics} and the relation $s(t)=-(1+r(t))\alpha(t)$ from \eqref{rs_exp} produce the following ODEs for $f$ and $g$:
\begin{align*}
f'(t)&=\beta(t)\big(1-f(t)-g(t)\big)+\alpha(t)g(t), \quad f(0)=0,\\
g'(t)
&=r(t)\beta(t)\big(1-f(t)-g(t)\big)-\alpha(t)g(t), \quad g(0)=0.
\end{align*}
We can find explicit expressions of the unique solution of the above ODE system by using  \eqref{equilibrium_explicit}:
\begin{align*}
f(t)&=1-\frac{(1+2r(t))^{\frac{3}{2}}}{r_0\sqrt{1+2r_0}(1+r(t))} + \frac{(1+r_0-r_0^2)(1+2r(t))}{r_0(1+2r_0)(1+r(t))},\\
g(t)&=\frac{1+2r(t)}{1+r(t)}\left( \frac{1+r(t)-r(t)^2}{r_0\sqrt{1+2r_0}\sqrt{1+2r(t)}}- \frac{1+r_0-r_0^2}{r_0(1+2r_0)}   \right).
\end{align*}
These expressions give 
\begin{equation}
\begin{split}\label{f''_boundary}
f''(0)&=-\frac{\sigma_w^4 r_0^4 }{\sigma_\ta^4 (1+3r_0)}<0,\\
f''(T-)&=\frac{3\sigma_w^4 r_0^3(1+2r_0)^2}{\sigma_\ta^4 (1+r_0)^4}\Big(\sqrt{1+2r_0}-1+r_0(r_0-1) \Big)>0,
\end{split}
\end{equation}  
where the second equality is due to $r(T-)=0$ and the second inequality is due to 
\begin{align*}
\begin{cases}
\left(\sqrt{1+2r_0}-1+r_0(r_0-1) \right)\Big|_{r_0=0}=0,\\
\frac{d}{dr_0}\left(\sqrt{1+2r_0}-1+r_0(r_0-1) \right)\Big|_{r_0=0}=0, \\
\frac{d^2}{dr_0^2}\left(\sqrt{1+2r_0}-1+r_0(r_0-1) \right)=2-\frac{1}{(1+2r_0)^{\frac{3}{2}}}>0, \quad \textrm{for} \quad r_0\in (0,\infty).
\end{cases}.
\end{align*}
Let $H:[0,\infty)^2\to \bR$ be defined as
\begin{align*}
H(x,y):=\left(\frac{\sigma_\ta^6 (1+r_0)^6 (1+2r(t))^\frac{7}{2} (1+3r(t))^5}{\sigma_w^6 r_0^5 (1+2r_0)^3 (1+r(t))^5}    f'''(t)\right)\Bigg|_{r(t)=x, \,\, r_0=y}.
\end{align*}
Direct computations produce for $0<x\leq y$:
\begin{align*}
\begin{cases}
H(x,x)=(1+x)(1+3x)^3(1+2x+4x^2)\sqrt{1+2x}>0,\\
H_y(x,x)= \frac{2(1+3x)(1+x(20+x(4+3x)(22+3x(7+4x))))}{\sqrt{1+2x}}>0,\\
H_{yy}(x,y)=\frac{2(1+3y(3+5y))(11+x(55+x(83+33x)))}{(1+2y)^\frac{3}{2}}>0,
\end{cases},
\end{align*}
where $H_y$ and $H_{yy}$ denote partial derivatives. These inequalities imply that 
\begin{align}
H(x,y)>0 \quad \textrm{for}\quad 0<x\leq y. \label{H_ineq}
\end{align}
Since $0< r(t)\leq r_0$ for $t\in [0,T)$, the definition of $H$ and \eqref{H_ineq} produce 
\begin{align*}
0<H(r(t),r_0)=\frac{\sigma_\ta^6 (1+r_0)^6 (1+2r(t))^\frac{7}{2} (1+3r(t))^5}{\sigma_w^6 r_0^5 (1+2r_0)^3 (1+r(t))^5}    f'''(t) \quad \textrm{for}\quad t\in [0,T),
\end{align*}
and we obtain
\begin{align}
f'''(t)>0\quad t\in [0,T). \label{f'_convex}
\end{align}
Combining \eqref{f''_boundary} and \eqref{f'_convex}, we conclude that the map $t\mapsto \frac{\bE[\theta'_t | \FF_0^I]}{\ta}=f'(t)$ is $U$ shaped for $t\in [0,T)$.

Finally, to prove $0<\frac{1}{\ta}\bE[\Delta \theta_T | \FF_0^I]<1$, we observe
\begin{align}
\begin{split}
\frac{1}{\ta}\bE[\Delta \theta_T | \FF_0^I] & = 1-f(T-)\\
&=\frac{\sqrt{1+2r_0}-1+r_0(r_0-1)}{r_0 (1+2r_0)},
\end{split}
\end{align}
where the first equality uses $\Delta \theta_T=\ta-\theta_{T-}$ and the definition of $f$, and the second equality  uses the explicit expression of $f$ and $r(T-)=0$. The conclusion follows because $\frac{\sqrt{1+2r_0}-1+r_0(r_0-1)}{r_0 (1+2r_0)}\in (0,1)$ for $r_0>0$.

(4): $\Delta\theta_T = Q_{T-}$ is from \eqref{predictable_jump} and $\Delta P_T=0$ is from \eqref{P_no_jump}. We obtain $Q_{T-}\neq 0$ a.s. because $\frac{1}{\ta}\bE[Q_{T-} | \FF_0^I]=\frac{1}{\ta}\bE[\Delta \theta_T | \FF_0^I]\neq 0$ by part (3).

(5): The explicit solution  of \eqref{Sigma4_ode} is given by
\begin{align}\label{notequal0}
\begin{split}
 \Sigma_4(t) =\frac{\rho^2 \sigma_\tv^2 \sqrt{1+2r_0}}{(1+r_0)^2} \frac{(1+r(t))^2}{\sqrt{1+2r(t)}}+ (1-\rho^2)\sigma_\tv^2.
\end{split}
\end{align}
Because $\tv-\rho\frac{\sigma_\tv}{\sigma_\ta}\ta$ is independent of $\FF_t^I$, we obtain for $t\in [0,T)$ that
\begin{align}
\bE \left[\left(\rho\tfrac{\sigma_\tv}{\sigma_\ta}\ta- P_t \right)^2\right] &= \bE \left[\left(\tv- P_t \right)^2\right]-\bE \left[\left(\tv-\rho\tfrac{\sigma_\tv}{\sigma_\ta}\ta\right)^2\right] \nonumber\\
&=\Sigma_4(t)-(1-\rho^2)\sigma_\tv^2 \nonumber\\
&=\frac{\rho^2 \sigma_\tv^2 \sqrt{1+2r_0}}{(1+r_0)^2} \frac{(1+r(t))^2}{\sqrt{1+2r(t)}}. \label{notequal}
\end{align}
The expression in \eqref{notequal} and Fatou's lemma produce 
$$
\bE \left[\left(\rho\frac{\sigma_\tv}{\sigma_\ta}\ta- P_T \right)^2\right]\ge \frac{\rho^2 \sigma_\tv^2 \sqrt{1+2r_0}}{(1+r_0)^2}>0,
$$ 
where we have used $r(T-)=0$ and $P_T=P_{T-}$.
\end{proof}


Figure~\ref{figure:properties} illustrates the price impact function $\lambda(t)$, the insider's expected order rates $\frac{1}{\ta}\bE[\theta'_t|\FF_0^I]$, the scaled autocorrelation of aggregate holdings, and the remaining (unconditional) variance of 
$$
P_T - \bE[\tv|\mathcal{F}_0^I] = P_T-\rho\frac{\sigma_\tv}{\sigma_\ta}\ta.
$$

\vspace{-1cm}
\begin{figure}[!h]
\begin{center}
	\caption{Graphs of the Kyle's $\lambda(t)$ (1A), the insider's expected order rate $\frac{1}{\ta}\bE[\theta'_t|\sigma(\ta)]$ (1B), autocorrelation of aggregate holdings $Y_t$ (1C), and remaining variance $\bE [(\rho\tfrac{\sigma_\tv}{\sigma_\ta}\ta- P_t )^2]$ (1D). The parameters are $\sigma_w:=1$, $\sigma_\tv:=1$, $\rho:=0.3$, $T:=1$, and $\sigma_\ta:=5$ (\textcolor{myblue}{\bf ---}), $\sigma_\ta:=3$ (\textcolor{myorange}{-\,-\,-}), and $\sigma_\ta:=1$ (\textcolor{mygreen}{- $\cdot$ -}).  }\ \\
\begin{footnotesize}
$\begin{array}{cc}
\includegraphics[width=6cm, height=4.5cm]{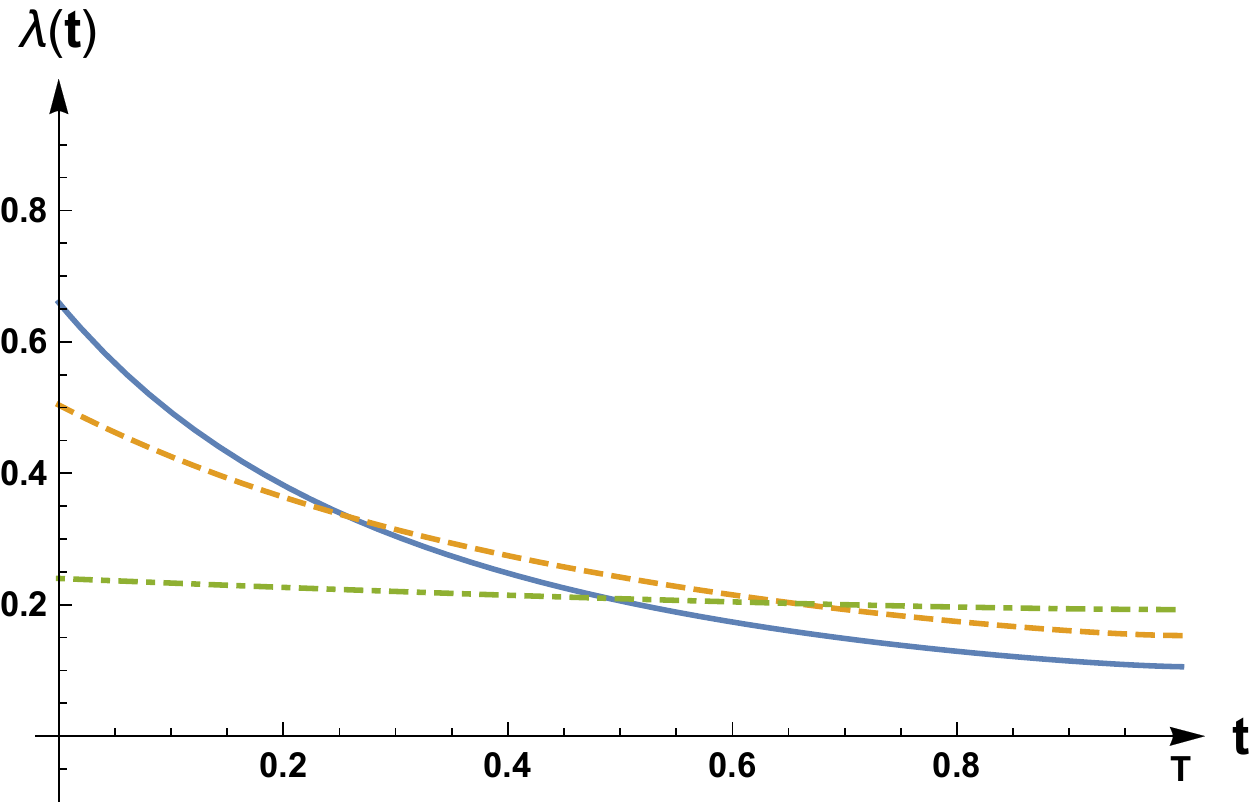} &\includegraphics[width=6cm, height=4.5cm]{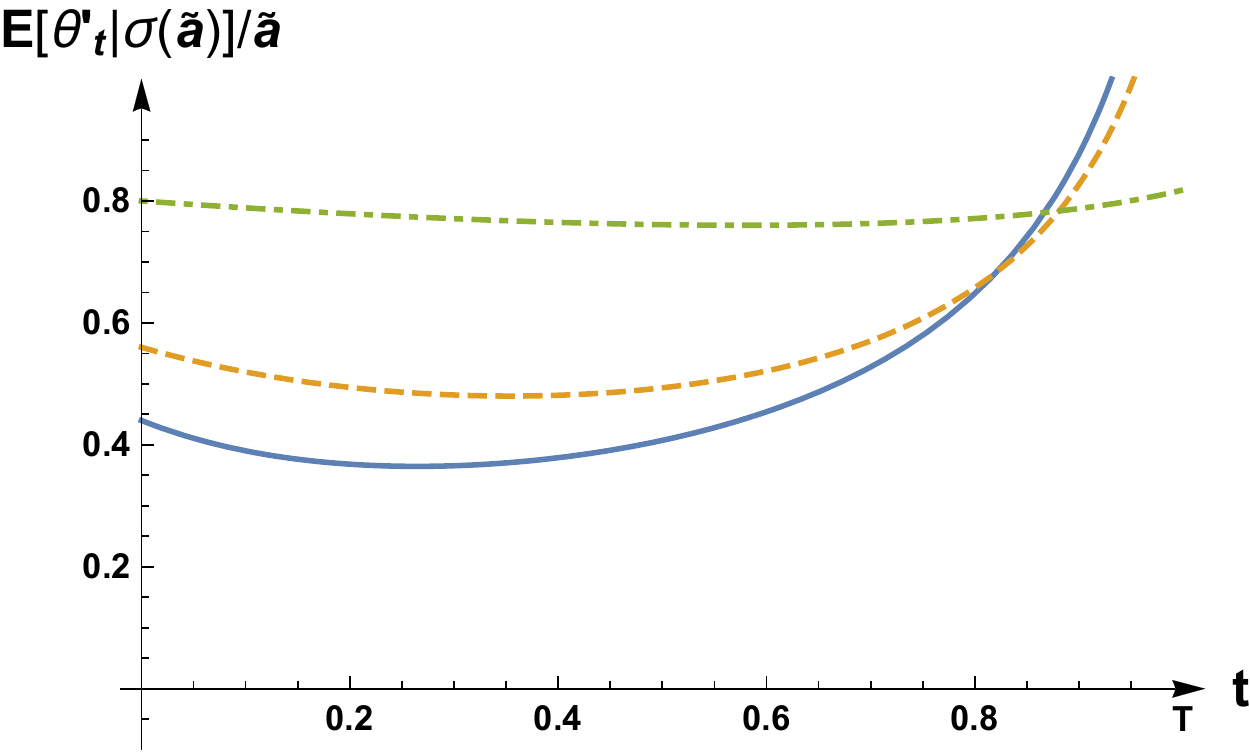} 
\\ 
\text{{\bf 1A:} price impact $\lambda(t)$ in \eqref{equilibrium_explicit}} &\text{\bf 1B:} \text{ expected order rate } \frac{1}{\ta}\bE[\theta'_t|\sigma(\ta)]
\vspace{0.1cm}
\\

\includegraphics[width=6cm, height=4.5cm]{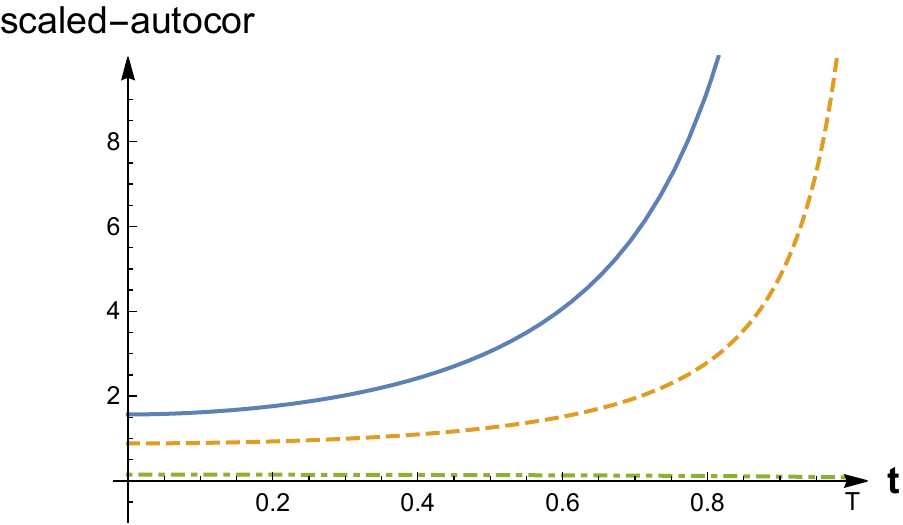} &\includegraphics[width=6cm, height=4.5cm]{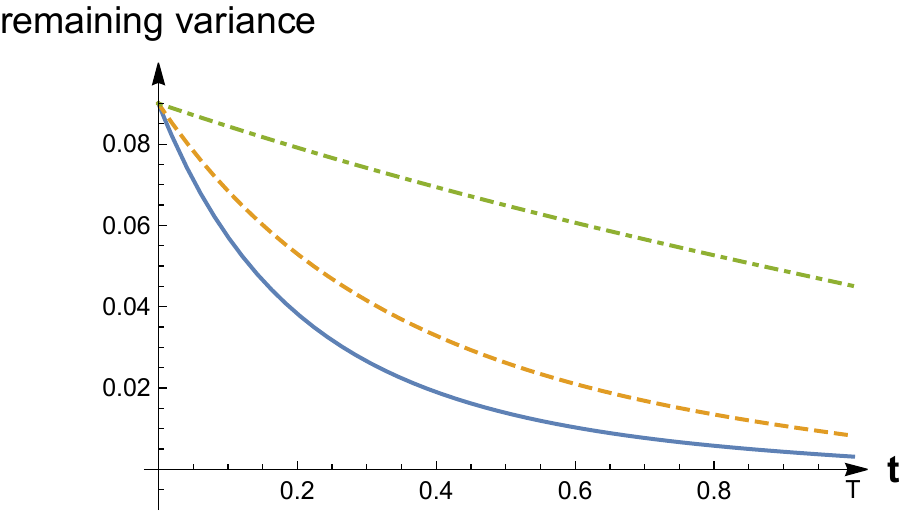} 
\\ 
\text{{\bf 1C:}  scaled autocorrelation in \eqref{scaledauto} } &\text{{\bf 1D:} remaining variance  $\bE [(\rho\tfrac{\sigma_\tv}{\sigma_\ta}\ta- P_t )^2]$ in \eqref{notequal}}\\

\end{array}$
 \end{footnotesize}
	\label{figure:properties}
\end{center}
\end{figure}

\appendix

\bigskip

\section{Kalman-Bucy filtering}\label{app:KB}

The following result is a special case of the classic Kalman-Bucy filtering result for Gaussian processes. 

\begin{lemma}[Kalman-Bucy] \label{filter_lem} Let the functions $(\Sigma_1,\lambda,\mu,r,s,\beta,\alpha)$ be as in \eqref{equilibrium_explicit}, let $\Sigma_2$ be as in \eqref{Sigma2}, and let the processes $(\theta_t, P_t, Q_t)$ be the solutions of the SDEs \eqref{q_dynamics}, \eqref{p_terminal}, and \eqref{theta_dynamics} for $t\in[0,T)$.
Then, for $t\in [0,T)$, the representations in \eqref{p_rational}, \eqref{q_rational}, and \eqref{Sigmas} hold.

\end{lemma}

\begin{proof} 
While $\Sigma_1(t)$ is given as the solution of the ODE in \eqref{S1_ode}, we can use $\lambda(t)$ in \eqref{equilibrium_explicit} to see that $\Sigma_2(t)$ defined in \eqref{Sigma2} satisfies the ODE in  \eqref{S3_ode}.

Based on \eqref{theta_dynamics}, the market makers' observation process $Y_t$ in \eqref{Y} has dynamics 
\begin{align}
dY_t =
\begin{cases}
& \sigma_wdW_t + \Big(\beta(t) (\ta -\theta_t-Q_t) + \alpha(t) Q_t\Big) dt,\quad t\in(0,T),\\
& \ta - \theta_{T-},\quad t=T.
\end{cases}
\end{align}
Because the explicit expressions in \eqref{equilibrium_explicit} satisfy  \eqref{lm_exp}-\eqref{rs_exp}, the Kalman-Bucy result (see, e.g., Theorem 10.3 in Liptser and Shiryaev 2001) gives \eqref{p_rational} and \eqref{q_rational} for $t\in [0,T)$. Furthermore, because $\Sigma_1$ satisfies \eqref{S1_ode} and $\Sigma_2$ satisfies \eqref{S3_ode}, the Kalman-Bucy result also gives \eqref{Sigmas}.
\end{proof}

\section{Full information}\label{app:full}
This appendix briefly explains why the equilibrium produced by Theorem \ref{main_thm} continues to be an equilibrium when \eqref{FI_t} is replaced with $\FF_t^I:=\sigma(\tv, \ta, (W_u)_{u\in [0,t]})$. Because $\FF_0^I=\sigma(\tv, \ta)$, the right-hand-side of the insider's maximization problem \eqref{rebalancer_problem} is altered and becomes 
\begin{align}\label{full_insider}
 \sup_{\theta \in \mathcal{A}} \mathbb{E} \left[ (\tilde{v}-P_T)\theta_T + \int_{[0,T]} \theta_{t-} dP_t \middle| \mathcal{F}^I_0\right] = \tilde{v}\tilde{a}-\inf_{\theta \in \mathcal{A}} \mathbb{E} \left[ \int_{[0,T]} (\tilde{a}-\theta_{t-})dP_t\middle| \mathcal{F}^I_0 \right],
 \end{align}
where the admissible set $\cA$ is as in Definition \ref{ad_def}. While a term like $\gamma(t)(\tv - P_t)$ for a deterministic function $\gamma(t)$ is crucial in Kyle (1985) and Back (1992), such a term is irrelevant in our constrained setting because the infimum in \eqref{full_insider} does not depend on $\tv$. Consequently,  the equilibrium in Theorem \ref{main_thm} remains a valid equilibrium even when the insider initially observes both $\ta$ and $\tv$.
\ \\

\end{document}